\documentclass[12pt]{article}

\usepackage{amsmath}
\usepackage{amssymb}
\usepackage{amsfonts}
\usepackage{latexsym}
\usepackage{color}
\usepackage{graphicx}
\usepackage[normalem]{ulem}

\catcode `\@=11 \@addtoreset{equation}{section}
\def\theequation{\arabic{section}.\arabic{equation}}
\catcode `\@=12

\renewcommand{\thesection}{\Roman{section}}

\newcommand{\be}{\begin{align}}
\newcommand{\en}{\end{align}}
\newcommand{\bea}{\begin{align}}
\newcommand{\ena}{\end{align}}
\newcommand{\beano}{\begin{align}}
\newcommand{\enano}{\end{align}}
\newcommand{\bee}{\begin{enumerate}}
\newcommand{\ene}{\end{enumerate}}

\newcommand{\mc}{\mathcal}

\newcommand{\E}{{\cal E}}
\newcommand{\V}{{\cal V}}
\newcommand{\F}{{\cal F}}

\newcommand{\Lc}{{\cal L}}

\newcommand{\1}{1 \!\! 1}

\newcommand{\Hil}{\mc H}

\newcommand{\cP}{\ensuremath{\mathcal{P}}}
\newcommand{\cT}{\ensuremath{\mathcal{T}}}
\newcommand{\cPT}{\ensuremath{\mathcal{PT}}}

\newtheorem{thm}{Theorem}

\newtheorem{prop}[thm]{Proposition}

\newenvironment{proof}{\noindent {\bf Proof --}}{\hfill$\square$ \vspace{3mm}\endtrivlist}

\renewcommand{\thesection}{\Roman{section}}
\catcode `\@=11 \@addtoreset{equation}{section}
\catcode `\@=12

\begin{document}

\thispagestyle{empty}

\vspace*{2cm}

\begin{center}
{\Large \bf $\cPT$-symmetric graphene under a magnetic field}   \vspace{2cm}\\

{\large Fabio Bagarello}\\
  Dipartimento di Energia, Ingegneria dell'Informazione e Modelli Matematici,\\
Facolt\`a di Ingegneria, Universit\`a di Palermo,\\ I-90128  Palermo, Italy\\
e-mail: fabio.bagarello@unipa.it\\
home page: www.unipa.it/fabio.bagarello

\vspace{2mm}


\vspace{2mm}

{\large Naomichi Hatano}\\
Institute of Industrial Science, University of Tokyo,\\ Komaba 4-6-1, Meguro, Tokyo 153-8505, Japan\\
e-mail: hatano@iis.u-tokyo.ac.jp\\

\end{center}

\vspace*{1cm}

\begin{abstract}
\noindent  We propose a $\cPT$-symmetrically deformed version of the graphene tight-binding model under a magnetic field.
We analyze the structure of the spectra and  the eigenvectors of the Hamiltonians around the $K$ and  $K'$ points, both in the $\cPT$-symmetric and $\cPT$-broken regions.
In particular we show that the presence of the deformation parameter $V$  produces several interesting consequences, including the asymmetry of the zero-energy states of the Hamiltonians and the breakdown of the completeness of the eigenvector sets.
We also discuss the biorthogonality of the eigenvectors, which {turns out to be} different in the $\cPT$-symmetric and $\cPT$-broken regions.

\end{abstract}

\vspace{2cm}


\vfill


\newpage

\section{Introduction}\label{sectintr}


Since its isolation on an adhesive tape~\cite{graphene}, graphene has quickly become a material of intensive attention.
Many researches have revealed various interesting aspects of the material; see e.g.\ Refs.~\cite{review1,review2,review3,review4} for reviews.
One of the most interesting features emerges particularly when we apply a magnetic field to it~\cite{mag1,mag2,mag3}.
The Landau levels due to the magnetic field form a structure different from the simple two-dimensional electron gas in that there are levels of zero energy and in that the non-zero energy levels are spaced not equally but proportionally to the square root of the level number.

\subsection{$\cPT$-symmetric non-Hermitian Hamiltonian}
\label{sec1.1}

In the present paper, we apply to graphene yet another ingredient of recent interest, namely the $\cPT$ symmetry~\cite{Bender98,Bender99,Bender07,Bender16}.
In order to attract attention of condensed-matter physicists, let us briefly describe the $\cPT$ symmetry here.
It refers to the parity-and-time symmetry of a Hamiltonian.

The simplest example of the $\cPT$-symmetric Hamiltonian may be the two-by-two matrix
\begin{align}\label{eq1.10}
H=\begin{pmatrix}
iV & g \\
g & -iV
\end{pmatrix},
\end{align}
which we can interpret as a two-site tight-binding model:
the two sites are coupled with a real coupling parameter $g$;
the first site has a complex potential $iV$, which can represent injection of particles from the environment, because the amplitude of the wave vector would increase in time as $e^{Vt/\hbar}$ if the site were isolated;
the potential $-iV$ of the second site can represent removal of the particles to the environment.
The $\cP$ operator swaps the first and second sites, which is represented by the linear operator
\begin{align}\label{eq1.20}
\cP=\begin{pmatrix}
0 & 1 \\
1 & 0
\end{pmatrix}.
\end{align}
The $\cT$ operator is complex conjugation, which is an anti-linear operator.
It is easy to confirm that the Hamiltonian $H$ in Eq.~\eqref{eq1.10} satisfies
\begin{align}\label{eq1.30}
(\cPT)H(\cPT)=H;
\end{align}
the parity operation $\cP$ swaps $iV$ and $-iV$ but the time operation $\cT$ switches them back to the original.
This is what we mean by the $\cPT$ symmetry of the Hamiltonian.

The Hamiltonian $H$ has the eigenvalues
\begin{align}\label{eq1.40}
E^{(\pm)}=\pm\sqrt{g^2-V^2},
\end{align}
which are real for $g\geq V$, although $H$ is non-Hermitian; $H^\dag\neq H$.
For $g<V$, on the other hand, the eigenvalues become complex (pure imaginary in this specific model).
This transition between real and complex eigenvalues physically means the following.
In the strong-coupling case $g\geq V$, the particles injected to the first site can flow abundantly into the second site, where they are removed at the same rate as the injection.
This constitutes a stationary state of a constant flow, which is indicated by the reality of the eigenvalues.
In the weak-coupling case $g<V$, on the other hand, the particles tend to build up in the first site, while they keep becoming scarcer in the second site.
This instability is indicated by the non-reality of the eigenvalues.
The first situation is often called the $\cPT$-symmetric phase, whereas the second one is the $\cPT$-broken phase.

At the transition point $g=V$, not only the two eigenvalues coalesce with each other, but the corresponding two eigenvectors become parallel.
This is therefore not the standard degeneracy, but often called an exceptional point~\cite{Kato52,Kato58,Kato66}, which has a huge literature recently, including experimental studies~\cite{Dembowski01,Klaiman08,Lefebvre09,Heiss12,Brandstetter14,Peng14}.
At the exceptional point, the eigenvectors are not complete and the Hamiltonian $H$ is not diagonalizable.
(In fact, non-Hermitian matrices are \textit{generally diagonalizable} except at the exceptional points.)

The transition at an exceptional point between the two phases indeed happens in a very wide class of $\cPT$-symmetric operators.
Suppose that a general $\cPT$-symmetric Hamiltonian has an eigenvector $\phi_n$ with an eigenvalue $E_n$:
\begin{align}\label{eq1.50}
H\phi_n=E_n\phi_n.
\end{align}
Inserting the symmetry relation~\eqref{eq1.30}, we have
\begin{align}\label{eq1.60}
H(\cPT)\phi_n=(\cPT)E_n\phi_n,
\end{align}
where we used the fact $(\cPT)^2=1$.
When the eigenvalue $E_n$ is real, the operator $\cPT$ passes it, yielding
\begin{align}\label{eq1.70}
H(\cPT)\phi_n=E_n(\cPT)\phi_n,
\end{align}
which means that $(\cPT)\phi_n$ is also an eigenvector with the same eigenvalue.
If we assume no degeneracy of the eigenvalue $E_n$ for simplicity, we conclude that $(\cPT)\phi_n\propto\phi_n$;
indeed, we can choose the phase of $\phi_n$ so that we can make $(\cPT)\phi_n=\phi_n$;
namely, the eigenvector is $\cPT$-symmetric.
This is what happens in the $\cPT$-symmetric phase.
When the eigenvalue $E_n$ is complex, on the other hand, we have, instead of Eq.~\eqref{eq1.70},
\begin{align}\label{eq1.80}
H(\cPT)\phi_n=\overline{E}_n(\cPT)\phi_n,
\end{align}
where $\overline{E}_n$ denotes the complex conjugate of $E_n$.
This means that we always have a complex-conjugate pair of eigenvalues $E_n$ and $\overline{E}_n$ with the eigenvectors $\phi_n$ and $(\cPT)\phi_n$;
each eigenvector is not $\cPT$-symmetric anymore in spite of the fact that the Hamiltonian is still $\cPT$-symmetric.
This is what happens in the $\cPT$-broken phase.
In typical situations including the example~\eqref{eq1.10}, two neighboring real eigenvalues in the $\cPT$-symmetric phase, as we tune system parameters, are attracted to each other, collide at the exceptional point, and then become a pair of complex-conjugate eigenvalues in the $\cPT$-broken phase, which repel each other; see Fig.~\ref{fig-exceptional point}.
\begin{figure}
\centering
\includegraphics[width=0.4\textwidth]{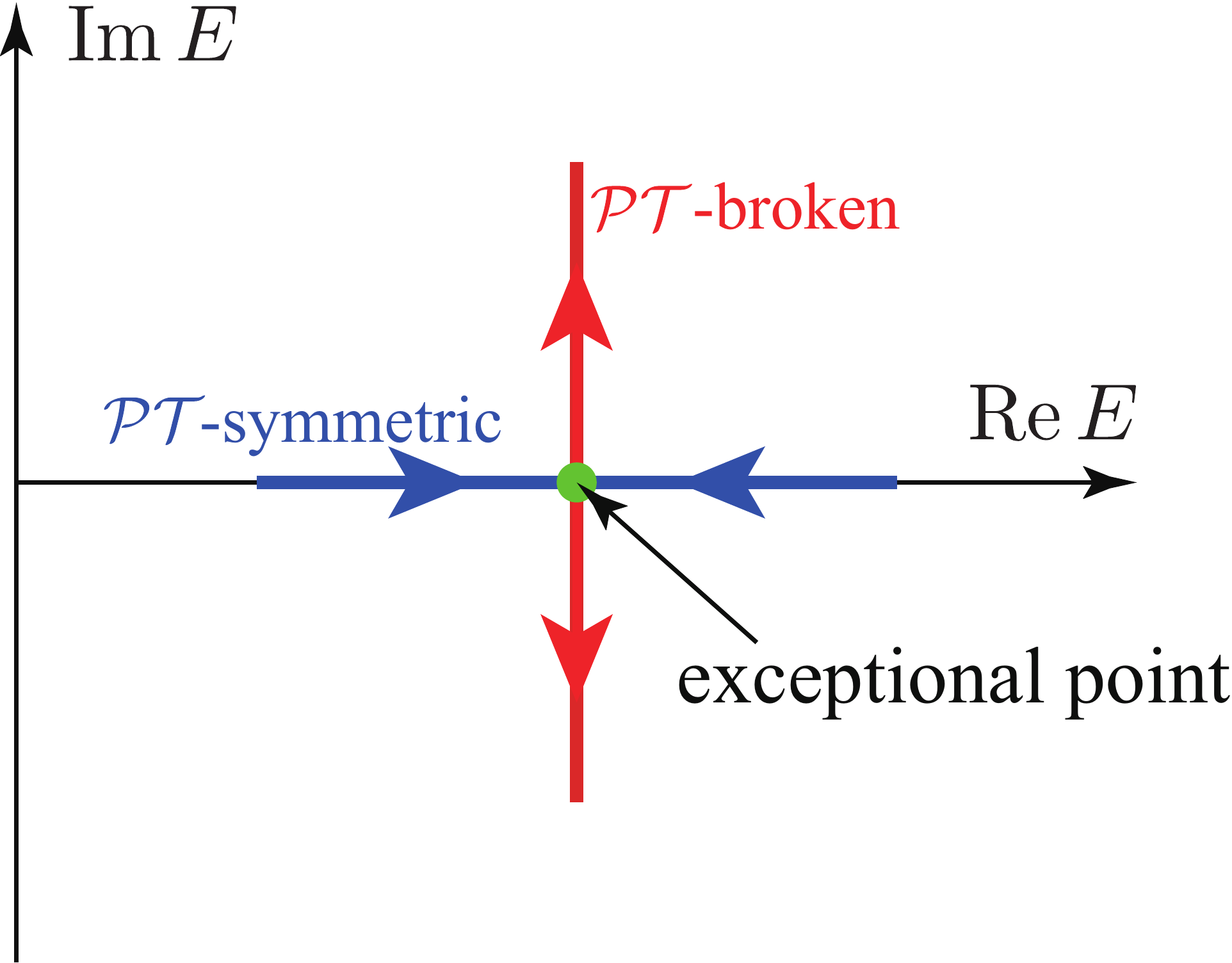}
\caption{(Color online) A schematic diagram of typical transition between the $\cPT$-symmetric and $\cPT$-broken phases.
As we tune system parameters from the $\cPT$-symmetric phase to the exceptional point and further on to the $\cPT$-broken phase, two real eigenvalues neighboring on the real axis are attracted to each other (indicated by the blue horizontal arrows on the real axis), collide at the exceptional point (indicated by a green dot), and become a complex-conjugate pair, which repel each other (indicated by the red vertical arrows).}
\label{fig-exceptional point}
\end{figure}

Questions of interest include the following:
Is it possible to formulate a standardized quantum mechanics for non-Hermitian but $\cPT$-symmetric Hamiltonians with real energy eigenvalues, namely in the $\cPT$-symmetric phase?
What is the general theoretical structure of the $\cPT$-broken phase, on the other hand?
A more specific subject of study is to find $\cPT$-symmetric models that describe physically interesting situations.

In the present paper, we introduce the potential $iV$ to one sublattice of graphene under a magnetic field and the potential $-iV$ to its other sublattice, which constitutes a $\cPT$-symmetric situation.
It is quite common to introduce a staggered chemical potential to graphene, that is, $\mu$ to one sublattice and $-\mu$ to the other sublattice,
which may be indeed realized by hexagonal lattice of boron-nitride~\cite{BN1,BN2,BN3,BN4}, in which boron atoms are on the A sublattice and nitride atoms are on the B sublattice.
Our $\cPT$-symmetric situation may be also realized in the following way:
suppose that we put a hexagonal lattice of two elements, such as boron-nitride, on a substrate;
assume that the substrate is an electron-doping material for one element but hole-doing for the other.
This can materialize our $\cPT$-symmetric situation.

For completeness, we should observe that few other $\cPT$, or non-Hermitian, versions of the graphene have been proposed in recent years, but in a different spirit with respect to ours \cite{gr1,gr2,gr3,gr4}.


\subsection{Brief overview of the graphene tight-binding model}

Let us now describe the model in more detail.
Graphene forms a hexagonal lattice, which is a bipartite lattice; see Fig.~\ref{fig2}.
\begin{figure}
\centering
\includegraphics[width=0.45\textwidth]{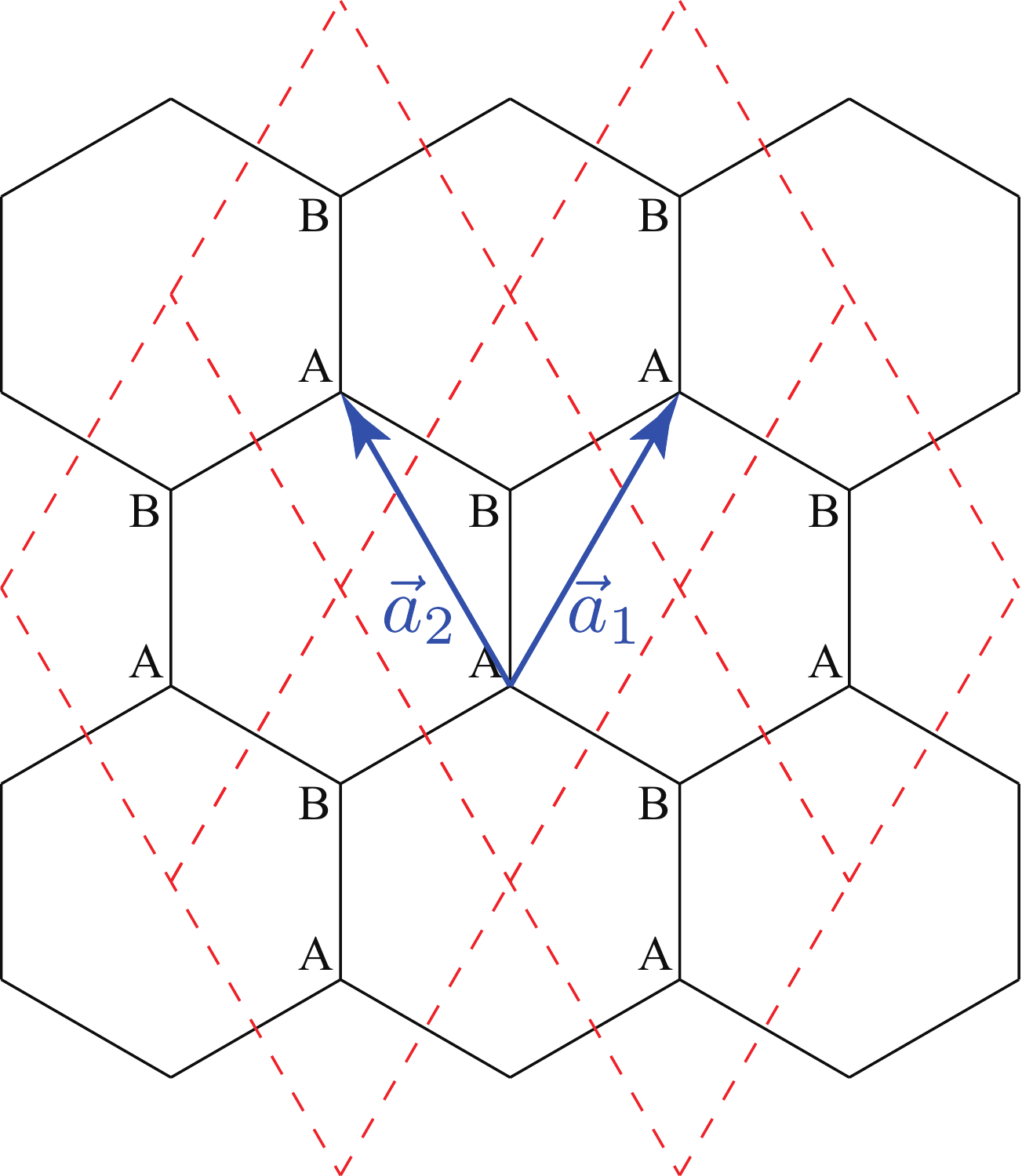}
\caption{(Color online) A hexagonal lattice (black solid lines). The red broken lines indicate unit cells.
Each unit cell consists of two sites, one on the A sublattice, the other on the B sublattice.
The two blue arrows indicate the vectors that denote the relative locations of the neighboring unit cells.}
\label{fig2}
\end{figure}
A unit cell, indicated by red broken lines in Fig.~\ref{fig2}, consists of two sites, one on the A sublattice and the other on the B sublattice.
If we assume that the electrons, specifically $\pi$ electrons, hop only to nearest-neighbor lattice points, those on a site on the A sublattice hop only to sites on the B sublattice, and vice versa.
The tight-binding Hamiltonian in the real-space representation therefore is of the following form:
on the diagonal, we have two-by-two blocks
\begin{align}\label{eq1.100}
H_\mathrm{unit}
=\begin{pmatrix}
\mu_A & t_1 \\
t_1 & \mu_B
\end{pmatrix},
\end{align}
which is the local Hamiltonian inside a unit cell with the chemical potentials $\mu_A$ and $\mu_B$ for the A and B sublattices, respectively, and with the {non-zero off-diagonal} intra-unit-cell hopping elements $t_1$ between the two sites.
In addition, the total Hamiltonian has the inter-unit-cell hopping elements $t_1$ between different unit cells.

By Fourier transforming the basis set with respect to the unit cells, we end up with the block-diagonalized Hamiltonian
\begin{align}\label{eq1.110}
H=\begin{pmatrix}
\ddots & 0 & 0 & 0 & 0 \\
0 & \tilde{H}_\mathrm{unit}(\vec{k}_1) & 0 & 0 & 0 \\
0 & 0 & \tilde{H}_\mathrm{unit}(\vec{k}_2) & 0 & 0 \\
0 & 0 & 0 & \tilde{H}_\mathrm{unit}(\vec{k}_3) & 0 \\
0 & 0 & 0 & 0 & \ddots
\end{pmatrix}
\end{align}
with
\begin{align}\label{eq1.120}
\tilde{H}_\mathrm{unit}(\vec{k})=\begin{pmatrix}
\mu_A & t_1(1+e^{i\vec{k}\cdot\vec{a}_1}+e^{i\vec{k}\cdot\vec{a}_2}) \\
 t_1(1+e^{-i\vec{k}\cdot\vec{a}_1}+e^{-i\vec{k}\cdot\vec{a}_2}) & \mu_B
\end{pmatrix},
\end{align}
where $\vec{a}_1$ and $\vec{a}_2$ are indicated in Fig.~\ref{fig2}.
We thereby have two energy eigenvalues for each wave number $\vec{k}$, which form two energy bands in the two-dimensional wave-number space, as shown in Fig.~\ref{fig-bands} for $\mu_A=\mu_B=0$.
\begin{figure}
\centering
\includegraphics[width=0.45\textwidth]{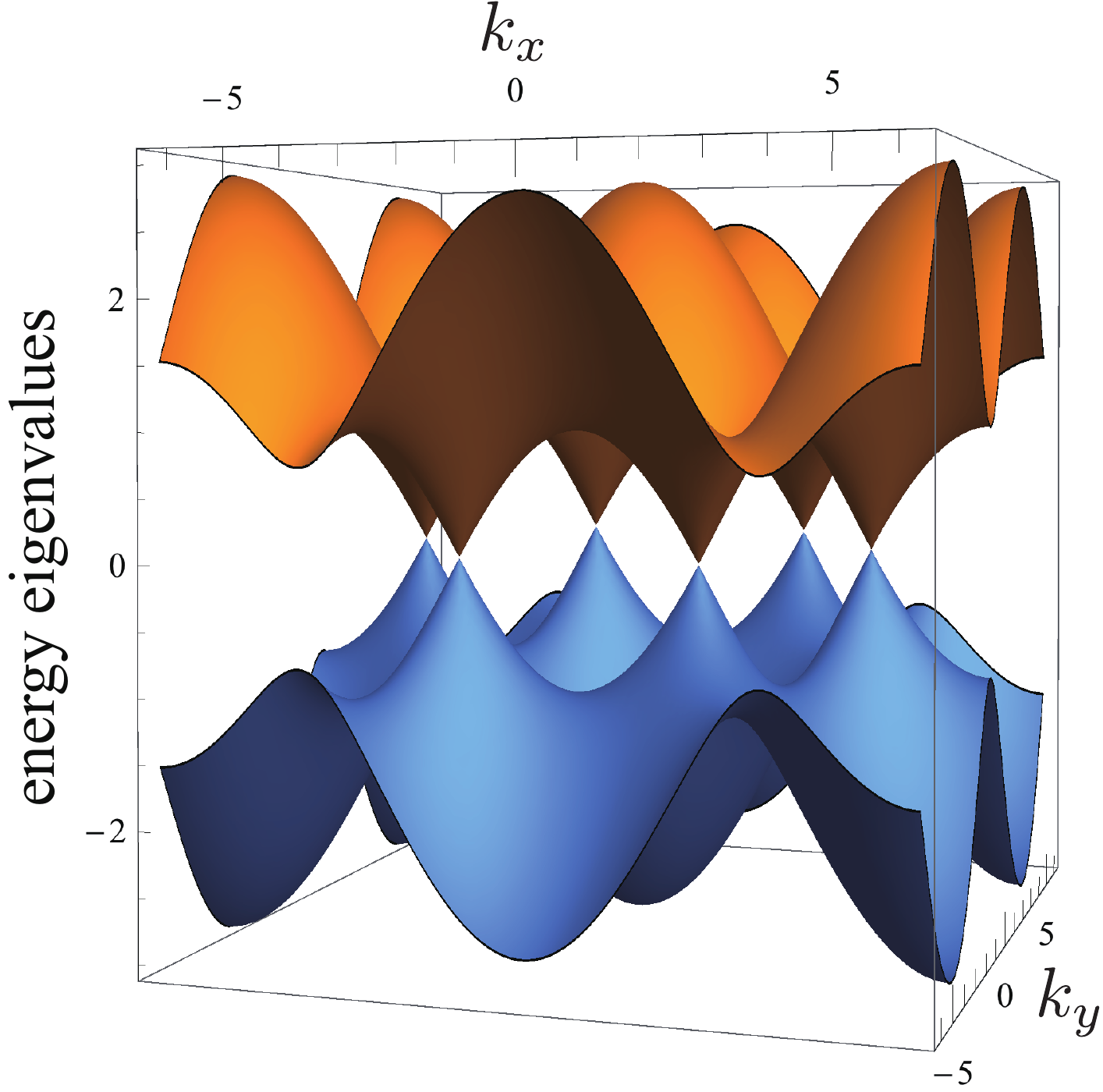}
\caption{(Color online) The energy bands of the tight-binding model in Fig.~\ref{fig2} with $\mu_A=\mu_B=0$.
The energy unit in the vertical axis is given by $t_1$, while the unit of the wave numbers $k_x$ and $k_y$ are given by the inverse of the lattice constant, which we put to unity here.}
\label{fig-bands}
\end{figure}

Among the blocks of the block-diagonalized Hamiltonian, the most important are the blocks of the two specific wave numbers, namely the Dirac points  $K$ and $K'$, respectively specified by
\begin{align}
\vec{K}=\frac{2\pi}{3}\begin{pmatrix}
1\\
\sqrt{3}
\end{pmatrix},
\qquad
\vec{K}'=\frac{2\pi}{3}\begin{pmatrix}
-1\\
\sqrt{3}
\end{pmatrix},
\end{align}
at which the energy eigenvalues are degenerate to zero for $\mu_A=\mu_B=0$.
Because the Fermi energy for graphene is zero, these points control the elementary excitation of graphene.
The upper and lower energy bands touch at these points, as can be seen in Fig.~\ref{fig-bands}, forming Dirac cones around the points, which are schematically shown in Fig.~\ref{fig4}(a).
\begin{figure}
\centering
\includegraphics[width=0.6\textwidth]{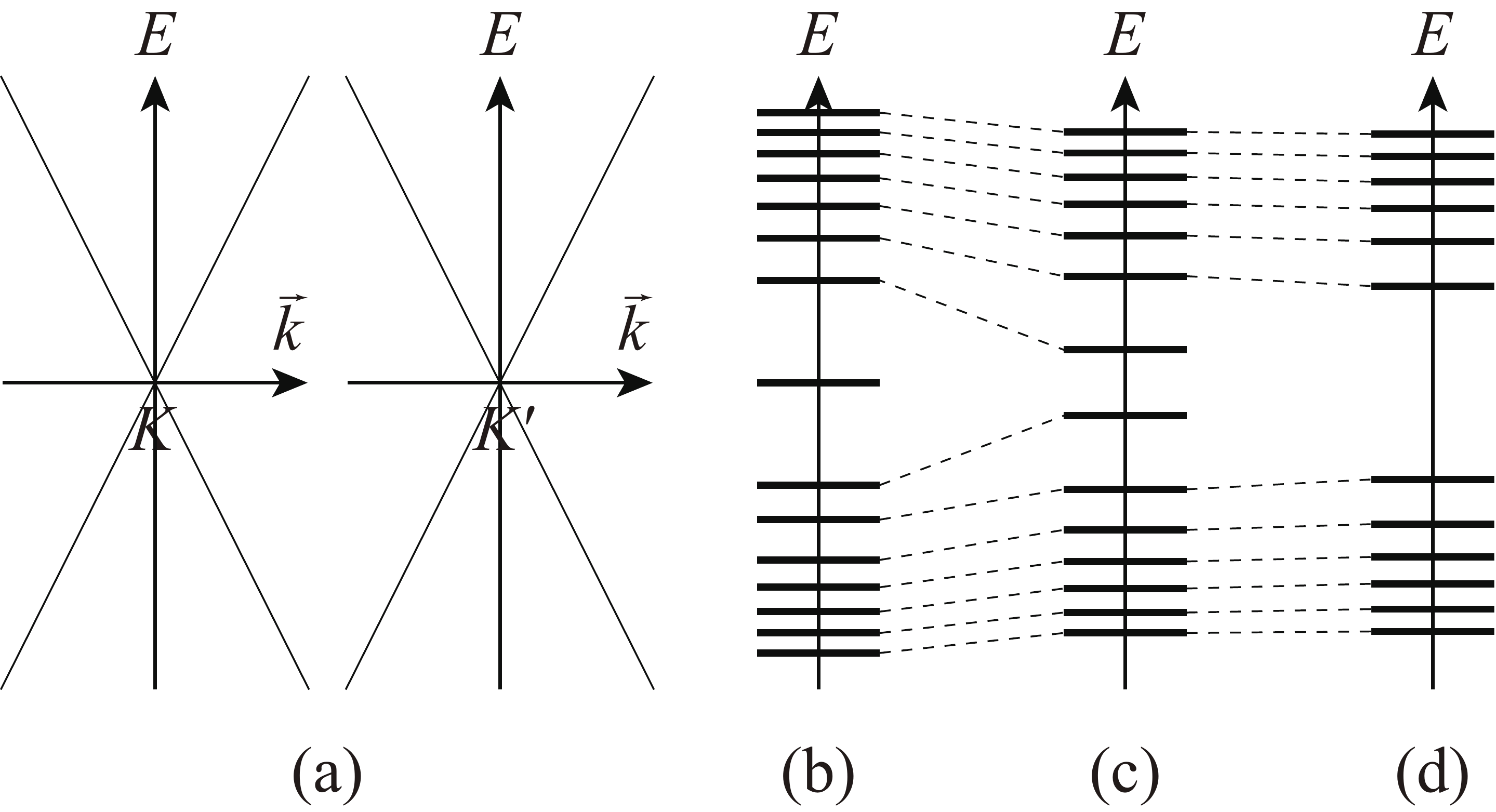}
\caption{(a) The dispersion relation of the Dirac cones around the $K$ and $K'$ points for $\mu_A=\mu_B=0$.
The Fermi energy of graphene is zero, which coincides with the Dirac points $K$ and $K'$.
(b) The Landau levels are formed under a magnetic field.
The levels are spaced proportionally to $\sqrt{n_2}$;
see Sec.~\ref{sectsam} for the definition of the quantum number $n_2=0,1,2,\ldots$.
(c) Shifts of the Landau levels for $V=0.9$.
The levels are spaced as $\sqrt{n_2-V^2}$.
The central Landau level $n_2=0$ has already become complex.
(d) Further shifts for $V=1.1$.
The levels with $n_2=1$ have collided with each other and become complex.}
\label{fig4}
\end{figure}
In the standard graphene, therefore, the low-energy excitations follow relativistic quantum mechanics;
this is one big feature of graphene, namely, the desktop relativity.

\subsection{Summary of the results}

As we predicted above, we apply two ingredients to the graphene tight-binding model, namely a magnetic field and a $\cPT$-symmetric chemical potential.
First, the spectrum is quantized to the Landau levels under a magnetic field.
Focusing on the Dirac cones around the $K$ and $K'$ points, we can write down the effective Hamiltonian as in Eq.~\eqref{20} below.
As is well studied (see e.g.\ Ref.~\cite{mag2}), which we will repeat in our way in Sec.~\ref{sectsam}, the Landau levels are not equally spaced as in the standard two-dimensional electron gas, but spaced proportionally to $\sqrt{n_2}$, as shown schematically in Fig.~\ref{fig4}(b);
see Sec.~\ref{sectsam} for the definition of the quantum number $n_2=0,1,2,\ldots$.
Each Landau level has an infinite number of degeneracy because of another quantum number $n_1=0,1,2,\ldots$.

We then further apply the $\cPT$-symmetric potential to the model.
We set the potentials to $\mu_A=iV$ for the A sublattice and $\mu_B=-iV$ for the B sublattice, as is represented in Eq.~\eqref{31} below.
Let us define the $\cP$ operation as the mirror reflection with respect to the horizontal axis of Fig.~\ref{fig2};
it then swaps the A and B sublattices with each other, changes the sign of the potentials $\pm iV$, which is represented by the transformation
\begin{align}
\begin{pmatrix}
0 & 1 \\
1 & 0
\end{pmatrix}
\tilde{H}_\mathrm{unit}(\vec{k})
\begin{pmatrix}
0 & 1 \\
1 & 0
\end{pmatrix}.
\end{align}
The $\cT$ operation, which is the complex conjugation, then changes the Hamiltonian back to the original one.
See the end of Sec.~\ref{sec1.1} for a possible materialization of the $\cPT$-symmetric situation.

We will show in Sec.~\ref{sec3} that the Landau levels are then spaced proportionally to $\sqrt{n_2-V^2}$ (after proper parameter normalization) under a set of biorthogonal eigenstates.
Therefore, as we increase the potential $V$, two Landau levels labeled by $n_2$ approach each other, collide with each other when $V^2=n_2$, which is an exceptional point, and then split into a pair of two pure imaginary eigenvalues $\pm i\sqrt{V^2-n_2}$; see Fig.~\ref{fig4}(c--d).
We will deduce that at this exceptional point, the eigenvectors of the two Landau levels become parallel, which makes the set of biorthogonal eigenstates incomplete.
Note that each Landau level still has an infinite number of degeneracy because of the other quantum number $n_1$.

Note also that the central level $n_2=0$ becomes a complex eigenvalue as soon as we introduce the $\cPT$-symmetric potential $V$; see Fig.~\ref{fig4}(b--c).
We show that the central level $n_2=0$ of the $K$ point coalesces with the central level $n_2=0$ of the $K'$ point and becomes complex as $\pm iV$.
This particular coalescence, however, is not an exceptional point but a degeneracy because it occurs in the Hermitian limit $V=0$.
These are probably the main results of the present paper.

This article is organized as follows: in the next Section~\ref{sectsam} we briefly review the Hermitian version of the model under a magnetic field and some of its main mathematical characteristics. In Sec.~\ref{sec3} we introduce our $\cPT$-symmetrically deformed version of the model with $\pm iV$ and consider the consequences of this deformation. Our conclusions and future perspective are given in Sec.~\ref{sec5}. { To make the paper self-contained, we have added \ref{appA} with some useful facts for non-Hermitian Hamiltonians.}

\section{The Dirac cones under a magnetic field}\label{sectsam}

Let us first consider a layer of graphene in an external constant magnetic field along $z$: $\vec B=B \hat e_3$, which can be deduced from $\vec B=\nabla\wedge\vec A$ with a vector potential in the symmetric gauge, $\vec A=\frac{B}{2}(-y,x,0)$. The Hamiltonian for the two Dirac points $K$ and $K'$ can be written as~\cite{review1}
\begin{align}
H_D=\begin{pmatrix}
        H_K & 0 \\
        0 & H_{K'}
      \end{pmatrix},
\end{align}
where, in the units $\hbar=c=1$, we have
\begin{align}
H_K=v_F\begin{pmatrix}
        0 & p_x-ip_y+\frac{eB}{2}(y+ix) \\
        p_x+ip_y+\frac{eB}{2}(y-ix) & 0
      \end{pmatrix},
\label{20}
\end{align}
while $H_{K'}$ is just its transpose: $H_{K'}=H_K^T$. Here $x,y,p_x$ and $p_y$ are the canonical, Hermitian, two-dimensional position and momentum operators, which satisfy $[x,p_x]=[y,p_y]=i\1$ with all the other commutators being zero, where $\1$ is the identity operator in the Hilbert space $\Hil:=\Lc^2(\Bbb R^2)$. The factor $v_F$ is the so-called Fermi velocity. The scalar product in $\Hil$ will be indicated as $\left\langle.,.\right\rangle$.

Let us now introduce the parameter called the magnetic length, $\xi=\sqrt{2/(e|B|)}$, as well as the following canonical operators:
\begin{align}
X=\frac{1}{\xi}x,\qquad Y=\frac{1}{\xi}y,\qquad P_X=\xi p_x, \qquad P_Y=\xi p_y.
\end{align}
These operators can be used to define two different pairs of bosonic operators: we first put $a_X=(X+iP_X)/\sqrt{2}$ and $a_Y=(Y+iP_Y)/\sqrt{2}$, and then
\begin{align}
A_1=\frac{a_X-ia_Y}{\sqrt{2}},\qquad A_2=\frac{a_X+ia_Y}{\sqrt{2}}.
\label{21}
\end{align}
The following commutation rules are satisfied:
\begin{align}
[a_X,a_X^\dagger]=[a_Y,a_Y^\dagger]=[A_1,A_1^\dagger]=[A_2,A_2^\dagger]=\1,
\label{22}
\end{align}
with the other commutators being zero. In terms of these operators, $H_K$ appears particularly simple. Indeed, we find
\begin{align}
H_K^{(+)}=\frac{2iv_F}{\xi}
      \begin{pmatrix}
        0 & A_2^\dagger \\
        -A_2 & 0
      \end{pmatrix},
\qquad
H_{K'}^{(+)}=\frac{2iv_F}{\xi}
      \begin{pmatrix}
        0 & -A_2 \\
        A_2^\dag & 0
      \end{pmatrix}
\label{23}
\end{align}
for $B>0$ and
\begin{align}
H_K^{(-)}=\frac{2iv_F}{\xi}
      \begin{pmatrix}
        0 & -A_1 \\
        A_1^\dag & 0
      \end{pmatrix},
\qquad
H_{K'}^{(-)}=\frac{2iv_F}{\xi}
      \begin{pmatrix}
        0 & A_1^\dag \\
        -A_1 & 0
      \end{pmatrix}
\label{23.5}
\end{align}
for $B<0$. Note that $H_K^{(+)}$ and $H_K^{(-)}$ are different expressions of the same Hamiltonian~\eqref{20}.
It is evident that $H_K=H_K^\dagger$, and a similar conclusion can also be deduced for $H_{K'}$. It is also clear that neither $H_K^{(+)}$ nor $H_{K'}^{(+)}$ depends on $A_1$ and $A_1^\dagger$, so that their eigenstates possess a manifest degeneracy. The same is true for $H_K^{(-)}$ nor $H_{K'}^{(-)}$, which do not depend on $A_2$ and $A_2^\dagger$. However, from now on, we will essentially concentrate on $H_{K}^{(+)}$ and $H_{K'}^{(+)}$, except for what is discussed in Appendix B. Most of what we are going to discuss from now on can be restated easily for $H_{K}^{(-)}$ and $H_{K'}^{(-)}$. For instance, the eigenvectors of $H_{K}^{(-)}$ could be found from those of $H_{K'}^{(+)}$, replacing the operators $A_1$ and $A_1^\dagger$ with $A_2$ and $A_2^\dagger$, and vice versa.

Now, let $e_{0,0}\in \Hil$ be the non-zero vacuum of $A_1$ and $A_2$: $A_1e_{0,0}=A_2e_{0,0}=0$. Then we introduce, as usual,
\begin{align}
e_{n_1,n_2}=\frac{1}{\sqrt{n_1!n_2!}}(A_1^\dagger)^{n_1}(A_2^\dagger)^{n_2}e_{0,0};
\label{24}\end{align}
the set $\E=\left\{e_{n_1,n_2},\, n_j\geq0\right\}$ is an orthonormal basis for $\Hil$, being the same as the one for a two-dimensional harmonic oscillator.

Rather than working in $\Hil$, in order to deal with $H_K^{(+)}$ it is convenient to work in a different Hilbert space, namely the direct sum of $\Hil$ with itself, $\Hil_2=\Hil\oplus\Hil$:
\begin{align}
\Hil_2=\left\{f=
                  \begin{pmatrix}
                    f_1 \\
                    f_2
                  \end{pmatrix},
                \quad f_1,f_2\in\Hil
\right\}.
\end{align}
In the new Hilbert space $\Hil_2$, the scalar product $\left\langle.,.\right\rangle_2$ is defined as
\begin{align}
\left\langle f,g\right\rangle_2:=\left\langle f_1,g_1\right\rangle+\left\langle f_2,g_2\right\rangle,
\label{25}
\end{align}
and the square norm is $\|f\|_2^2=\|f_1\|^2+\|f_2\|^2$, for all $f=
                  \begin{pmatrix}
                    f_1 \\
                    f_2
                  \end{pmatrix}$, $g=
                  \begin{pmatrix}
                    f_1 \\
                    f_2
                  \end{pmatrix}$ in $\Hil_2$. Introducing now the vectors
\begin{align}
e_{n_1,n_2}^{(1)}=
                  \begin{pmatrix}
                    e_{n_1,n_2} \\
                    0
                  \end{pmatrix},
              \qquad e_{n_1,n_2}^{(2)}=
                  \begin{pmatrix}
                    0 \\
                    e_{n_1,n_2}
                  \end{pmatrix},
\label{26}
\end{align}
we have an orthonormal basis set $\E_2:=\{e_{n_1,n_2}^{(k)},\,n_1,n_2\geq0,\,k=1,2\}$ for $\Hil_2$. This means, among other things, that $\E_2$ is complete in $\Hil_2$: the only vector $f\in\Hil_2$ which is orthogonal to all the vectors of $\E_2$ is the zero vector.

In view of application to graphene it is more convenient to use a different orthonormal basis of $\Hil_2$, the set $\V_2=\{v_{n_1,n_2}^{(k)},\,n_1,n_2\geq0,\,k=\pm\}$, where
\begin{align}
v_{n_1,0}^{(+)}=v_{n_1,0}^{(-)}=e_{n_1,0}^{(1)}=
                  \begin{pmatrix}
                    e_{n_1,0} \\
                    0
                  \end{pmatrix},
\label{27}
\end{align}
Quite often, in the rest of the paper, we call this vector simply $v_{n_1,0}$. For $n_2\geq1$, we have
\begin{align}
v_{n_1,n_2}^{(\pm)}=\frac{1}{\sqrt{2}}
                  \begin{pmatrix}
                    e_{n_1,n_2} \\
                    \mp i e_{n_1,n_2-1}
                  \end{pmatrix}
               =\frac{1}{\sqrt{2}}\left(e_{n_1,n_2}^{(1)}\mp i e_{n_1,n_2-1}^{(2)}\right).
\label{28}
\end{align}
It is easy to check that these vectors are mutually orthogonal, normalized in $\Hil_2$, and complete. Hence, $\V_2$ is an orthonormal basis, as stated before. This is not surprising, since its vectors are indeed the eigenvectors of $H_K^{(+)}$:
\begin{align}
H_K^{(+)} v_{n_1,0}=0, \quad H_K^{(+)} v_{n_1,n_2}^{(+)}=E_{n_1,n_2}^{(+)}  v_{n_1,n_2}^{(+)}, \quad   H_K^{(+)} v_{n_1,n_2}^{(-)}=E_{n_1,n_2}^{(-)}  v_{n_1,n_2}^{(-)},
\label{29}
\end{align}
where $E_{n_1,n_2}^{(\pm)}=\pm (2v_F/\xi)\sqrt{n_2}$. More compactly we can simply write $H_K^{(+)} v_{n_1,n_2}^{(\pm)}=E_{n_1,n_2}^{(\pm)} v_{n_1,n_2}^{(\pm)}$. We see explicitly that the eigenvalues have an infinite degeneracy with respect to the quantum number $n_1$, which can be removed by using the angular momentum~\cite{ang}.  We will not consider this aspect here, since it is not relevant for us.

Of course, both $\E_2$ and $\V_2$ can be used to produce two different resolutions of the identity. Indeed we have
\begin{align}
\sum_{n_1,n_2=0}^\infty \sum_{k=1}^2\left\langle e_{n_1,n_2}^{(k)},f\right\rangle_2 e_{n_1,n_2}^{(k)}=\sum_{n_1,n_2=0}^\infty \sum_{k=\pm}\left\langle v_{n_1,n_2}^{(k)},f\right\rangle_2 v_{n_1,n_2}^{(k)}=f,
\label{210}
\end{align}
for all $f\in\Hil_2$.

\vspace{2mm}

{\bf Remark:} What we have seen so far can be easily adapted to the analysis of the Hamiltonian for the other Dirac cone, $H_{K'}^{(+)}$, which is simply the transpose of $H_K^{(+)}$, and, as we have already pointed out, also to $H_K^{(-)}$ and $H_{K'}^{(-)}$.
We will say more on the other Dirac cone in Sec.~\ref{sec3.2}, { in the presence of the $\cPT$-symmetric potential}.

\section{$\cPT$-symmetric chemical potential}
\label{sec3}

We now introduce the $\cPT$-symmetric chemical potential to Eq.~(\ref{23}) as follows:
\begin{align}
H_K^{(+)}(V)=\frac{2iv_F}{\xi}
      \begin{pmatrix}
        V & A_2^\dagger \\
        -A_2 & -V
      \end{pmatrix},
\label{31}
\end{align}
where $V$ is assumed to be a strictly positive (real) quantity.
As we show the details in \ref{appB}, for $V=0$, the \textit{set} of Dirac cones at $K$ and $K'$ are time-reversal symmetric as well as parity symmetric.
For $V\neq 0$, it observes neither symmetries but does the $\cPT$ symmetry.

An easy extension of the standard arguments allows us to deduce that the general expression of the eigenvectors are still, as in the case with $V=0$, of the form (\ref{28}), but with some essential difference, which is also reflected in the form of the eigenvalues. In particular we first find that
\begin{align}
E_{n_1,n_2}^{(\pm)}=\frac{\pm 2v_F}{\xi}\,\sqrt{n_2-V^2},
\label{32}
\end{align}
which reduces to the known value if $V\rightarrow 0$, and which is still independent of $n_1$. A major difference appears as follows: if $n_2> V^2$, then the values of the energy are real; we are in the $\cPT$-symmetric region. As soon as $n_2<V^2$, however, the energy turns out to be complex, and we are in the $\cPT$-broken region. We will come back to this later on.

Going now to the eigenvectors, we first observe that
\begin{align}
\Phi_{n_1,0}^{(+)}=
                  \begin{pmatrix}
                    e_{n_1,0} \\
                    0
                  \end{pmatrix},
\label{33}
\end{align}
is an eigenvector of $H_K^{(+)}(V)$ with the eigenvalue $E_{n_1,0}^{(+)}=2iv_FV/\xi$. On the other hand, we can prove that there is no non-zero eigenstate corresponding to $E_{n_1,0}^{(-)}=-2iv_FV/\xi$. In fact, if we assume that such a non-zero vector $\Phi_{n_1,0}^{(-)}=
                  \begin{pmatrix}
                    c_1 \\
                    c_2
                  \end{pmatrix}
                $ does exist, it must satisfy the equation $H_K^{(+)}(V)\Phi_{n_1,0}^{(-)}=\frac{- 2iv_F}{\xi}\Phi_{n_1,0}^{(-)}$, which implies in turn that $c_1$ and $c_2$ should satisfy the equations $A_2c_1=0$ and $A_2^\dagger c_2=-2Vc_1$. Hence, acting on this last with $A_2$ and using the first, we obtain $A_2A_2^\dagger c_2=0$, so that $\|A_2^\dagger c_2\|=0$ and therefore $A_2^\dagger c_2=0$. Then we have
\begin{align}
0=A_2\left(A_2^\dagger c_2\right)=\left(\1+A_2^\dagger A_2\right)c_2\quad\Rightarrow \quad -\|c_2\|^2=\|A_2c_2\|^2.
\end{align}
For this last equality to be satisfied, we must have $\|c_2\|=\|A_2c_2\|=0$. Hence $c_2$ must be zero. The fact that $c_1=0$ also is now a consequence of the equality above $A_2^\dagger c_2=-2Vc_1$, at least if $V\neq 0$. Then the trivial vector $\Phi_{n_1,0}^{(-)}=
                  \begin{pmatrix}
                    0 \\
                    0
                  \end{pmatrix}
               $ is the only solution that satisfies the equation $H_K^{(+)}(V)\Phi_{n_1,0}^{(-)}=\frac{- 2iv_F}{\xi}\Phi_{n_1,0}^{(-)}$. In the limit $V=0$, on the other hand, the equation $A_2^\dagger c_2=-2Vc_1$ does not imply that $c_1=0$ and in fact a nontrivial ground state in this case does exist, as discussed in Sec.~\ref{sectsam}. The reason for this is that, if $V=0$, there is no difference between $E_{n_1,0}^{(+)}$ and $E_{n_1,0}^{(-)}$, which are both zero.

As for the levels with $n_2\geq1$, the normalized eigenstates are deformed versions of those in Eq.~(\ref{28}). More in detail, defining the following quantities, which are in general complex,
\begin{align}
\alpha_{n_1,n_2}^{(\pm)}=\mp i\frac{\sqrt{n_2-V^2}\mp iV}{\sqrt{n_2}},
\label{34}
\end{align}
we can write
\begin{align}
\Phi_{n_1,n_2}^{(\pm)}=\frac{1}{\sqrt{1+|\alpha_{n_1,n_2}^{(\pm)}|^2}}
                  \begin{pmatrix}
                    e_{n_1,n_2} \\
                    \alpha_{n_1,n_2}^{(\pm)} e_{n_1,n_2-1}
                  \end{pmatrix}.
\label{35}
\end{align}
With these definitions  we have
\begin{align}
H_K^{(+)}(V)\Phi_{n_1,n_2}^{(\pm)}=E_{n_1,n_2}^{(\pm)}\Phi_{n_1,n_2}^{(\pm)}.
\label{36}
\end{align}

It is easy to see what happens for ${H_K^{(+)}}^\dagger(V)$, since this can be recovered from $H_K^{(+)}(V)$ just replacing everywhere $V$ with $-V$. In particular, since the eigenvalues are quadratic in $V$, $H_K^{(+)}(V)$ and ${H_K^{(+)}}^\dagger(V)$ turn out to be isospectral. Concerning the eigenstates, these are deduced from the eigenvectors $\Phi_{n_1,n_2}^{(\pm)}$ just with the same substitution. More in details, calling
\begin{align}
\beta_{n_1,n_2}^{(\pm)}=\mp i\frac{\sqrt{n_2-V^2}\pm iV}{\sqrt{n_2}},
\label{37}
\end{align}
for all $n_2\geq1$, we can write
\begin{align}
\Psi_{n_1,n_2}^{(\pm)}=\frac{1}{\sqrt{1+|\beta_{n_1,n_2}^{(\pm)}|^2}}\left(
                  \begin{array}{c}
                    e_{n_1,n_2} \\
                    \beta_{n_1,n_2}^{(\pm)} e_{n_1,n_2-1} \\
                  \end{array}
                \right)
\label{38}
\end{align}
and
\begin{align}
{H_K^{(+)}}^\dagger(V)\Psi_{n_1,n_2}^{(\pm)}=E_{n_1,n_2}^{(\pm)}\Psi_{n_1,n_2}^{(\pm)}.
\label{39}
\end{align}

Analogously to what happens for $H_K^{(+)}(V)$, only one ground state of ${H_K^{(+)}}^\dagger(V)$ does exist, which coincides with $\Phi_{n_1,0}^{(+)}$ above. However, the corresponding eigenvalue is now $E_{n_1,0}^{(-)}=-2iv_FV/\xi$, so that we conclude that $\Psi_{n_1,0}^{(-)}=\Phi_{n_1,0}^{(+)}$. On the other hand, no non-zero eigenvector does exist which corresponds to  $E_{n_1,0}^{(+)}=2iv_FV/\xi$. We thus deduce a similar situation with respect to the one observed for $H_K^{(+)}(V)$. We therefore conclude that the case $n_2=0$ is really exceptional; indeed we have $\Psi_{n_1,0}^{(-)}=\Phi_{n_1,0}^{(+)}=
                  \begin{pmatrix}
                    e_{n_1,0} \\
                    0
                  \end{pmatrix}$, while neither $\Psi_{n_1,0}^{(+)}$ nor $\Phi_{n_1,0}^{(-)}$ do exist. { This should be remembered in the rest of the paper, since} all the formulas considered from now on, and in particular those in Section \ref{sec3.1}, are valid only when these particular vectors are not involved.

\vspace{2mm}

{\bf Remarks:} (i) In the limit $V\rightarrow0$, all the result reduces to the ones discussed in Sec.~\ref{sectsam}. In particular the fact that  $\Psi_{n_1,0}^{(-)}=\Phi_{n_1,0}^{(+)}$ agrees with the fact that, in this limit,  $E_{n_1,0}^{(-)}=E_{n_1,0}^{(+)}=0$. It is also interesting to observe that the coefficients $\alpha_{n_1,n_2}^{(\pm)}$ and $\beta_{n_1,n_2}^{(\pm)}$ simply returns $+i$ or $-i$, as in formula (\ref{28}).

{
(ii) The choice of normalization in (\ref{35}) and (\ref{38}) is such that $\|\Phi_{n_1,n_2}^{(\pm)}\| = \| \Psi_{n_1,n_2}^{(\pm)}\|=1$. We prefer this choice, rather than the one which could also be used which makes the scalar product between $\Phi_{n_1,n_2}^{(\pm)}$ and $\Psi_{n_1,n_2}^{(\pm)}$ equal to unity, since this biorthogonality strongly refer to the value of $V$. This will be evident in the next section.
}

\vspace{2mm}

\subsection{Biorthogonality of the eigenvectors}
\label{sec3.1}

Let us call $\F_\Phi=\{\Phi_{n_1,n_2}^{(j)},\,n_1, n_2\geq0, j=\pm\}$ and $\F_\Psi=\{\Psi_{n_1,n_2}^{(j)},\,n_1, n_2\geq0, j=\pm\}$. Because of their particular forms and because of the orthogonality of the vectors $e_{n_1,n_2}$, it is clear that
\begin{align}
\left\langle\Phi_{n_1,n_2}^{(j)},\Phi_{m_1,m_2}^{(k)}\right\rangle_2=\left\langle\Psi_{n_1,n_2}^{(j)},\Psi_{m_1,m_2}^{(k)}\right\rangle_2=0
\end{align}
for all $(n_1,n_2)\neq(m_1,m_2)$, and for all choices of $j$ and $k$. It is also possible to check that,
\begin{align}
\left\langle\Phi_{n_1,n_2}^{(+)},\Phi_{n_1,n_2}^{(-)}\right\rangle_2\neq0 \quad \mbox{and}\quad
\left\langle\Psi_{n_1,n_2}^{(+)},\Psi_{n_1,n_2}^{(-)}\right\rangle_2\neq 0.
\end{align}
Therefore, eigenstates of $H_K^{(+)}(V)$ corresponding to different eigenvalues are not mutually orthogonal. This is not surprising, since $H_K^{(+)}(V)$ is not Hermitian in the present settings. However, we can check that the orthogonality is recovered when $V$ is sent to zero, i.e., when $H_K^{(+)}(V)$ becomes Hermitian.

What still remains, as quite often in situations like ours, is the possible biorthogonality of the sets $\F_\Phi$ and $\F_\Psi$. In fact, this is not so automatic, and needs some care. The point is the following: if $H$ is not Hermitian but two of its eigenvalues $E_1$ and $E_2$ are real, then the states $\varphi_1$ and $\Psi_2$ that satisfy $H\varphi_1=E_1\varphi_1$ and $H^\dagger \Psi_2=E_2\Psi_2$ are guaranteed to be mutually orthogonal. If $E_1$ or $E_2$, or both, are complex, on the other hand, this is no longer granted in general.
We will show that in our particular situation of the $\cPT$-broken region, the biorthogonality of the sets $\F_\Phi$ and $\F_\Psi$ is recovered  { only when properly pairing the } eigenstates.

First we observe that $\left\langle\Phi_{n_1,n_2}^{(j)},\Psi_{m_1,m_2}^{(k)}\right\rangle_2$ can only be different from zero if $(n_1,n_2)=(m_1,m_2)$. Otherwise these scalar products are all zero. Now, if we compute
$
\left\langle\Phi_{n_1,n_2}^{(+)},\Psi_{n_1,n_2}^{(-)}\right\rangle_2
$ for instance,
we deduce that, neglecting an unnecessary multiplication factor,
\begin{align}
\left\langle\Phi_{n_1,n_2}^{(+)},\Psi_{n_1,n_2}^{(-)}\right\rangle_2\simeq 1+\overline{\alpha_{n_1,n_2}^{(+)}}\,\beta_{n_1,n_2}^{(-)}.
\end{align}
The result of this computation depends on the values of $n_2$ and $V$. In fact, we can check that for $n_2> V^2$, we have $\overline{\alpha_{n_1,n_2}^{(+)}}\,\beta_{n_1,n_2}^{(-)}=-1$, but for $n_2< V^2$ this is not true. Hence
\begin{align}
\left\langle\Phi_{n_1,n_2}^{(+)},\Psi_{n_1,n_2}^{(-)}\right\rangle_2\,
\begin{cases}
= 0, \qquad &\mbox{if } n_2> V^2,  \\
\neq 0, \qquad &\mbox{if } n_2< V^2.
\end{cases}
\end{align}
Similarly we can check that $\left\langle\Phi_{n_1,n_2}^{(-)},\Psi_{n_1,n_2}^{(+)}\right\rangle_2$ is zero for $n_2> V^2$, but is not zero otherwise.

It is also interesting to notice that a completely opposite result is deduced in the $\cPT$-broken region, i.e.\ for purely imaginary eigenvalues. In fact, for $n_2<V^2$, we deduce that the different pair satisfies
\begin{align}
\left\langle\Phi_{n_1,n_2}^{(\pm)},\Psi_{n_1,n_2}^{(\pm)}\right\rangle_2=0,
\end{align}
so that they are biorthogonal, while they are in general not for $n_2>V^2$: $\left\langle\Phi_{n_1,n_2}^{(\pm)},\Psi_{n_1,n_2}^{(\pm)}\right\rangle_2\neq0$ for $n_2>V^2$.

These results are of course related to the reality of the eigenvalues of $H_K^{(+)}(V)$ and ${H_K^{(+)}}^\dagger(V)$. In fact, when $n_2> V^2$, the eigenvalues $E_{n_1,n_2}^{(\pm)}$ are all real, and we know for general reasons that $\Phi_{n_1,n_2}^{(\pm)}$ must be orthogonal to $\Psi_{n_1,n_2}^{(\mp)}$, but not, in general, to $\Psi_{n_1,n_2}^{(\pm)}$. On the other hand, when the eigenvalues are purely imaginary,  $\Phi_{n_1,n_2}^{(\pm)}$ are necessarily orthogonal to  $\Psi_{n_1,n_2}^{(\pm)}$, but not to  $\Psi_{n_1,n_2}^{(\mp)}$.

This has consequences on the possibility of introducing a metric, at least in the way which is discussed in Ref.~\cite{bagbook} for instance; see \ref{appA}. Here, in fact, the intertwining operator between $H_K^{(+)}(V)$ and ${H_K^{(+)}}^\dagger(V)$ has the formal expression $S_\Phi=\sum_{n_1=0, n_2=0}^\infty\sum_{j=\pm}|\Phi_{n_1,n_2}^{(j)}\left\rangle \right\langle \Phi_{n_1,n_2}^{(j)}|$ (we recall that $\Phi_{n_1,0}^{(-)}=0$). However this operator acts in different ways depending on whether we are in the $\cPT$-symmetric or $\cPT$-broken region. For instance, for $n_2>V^2$ ($\cPT$-symmetric region),  $S_\Phi\Psi_{n_1,n_2}^{(+)}$ is proportional to $\Phi_{n_1,n_2}^{(+)}$. However, for $n_2<V^2$ ($\cPT$-broken region), we can find that $S_\Phi\Psi_{n_1,n_2}^{(+)}$ is proportional to $\Phi_{n_1,n_2}^{(-)}$. Therefore, except for some multiplicative coefficients which can be fixed properly, $S_\Phi$ can change eigenstates $\Psi_{n_1,n_2}^{(+)}$ of ${H_K^{(+)}}^\dagger(V)$ either into the eigenstates $\Phi_{n_1,n_2}^{(+)}$ or into the eigenstates $\Phi_{n_1,n_2}^{(-)}$ of $H_K^{(+)}(V)$, depending on the parameter region. Of course, a similar behavior is expected for an operator $S_\Psi$ defined in analogy with $S_\Phi$.

\vspace{2mm}

An interesting issue to consider now is the completeness of the sets  $\F_\Phi$ and $\F_\Psi$. In many applications in quantum mechanics with non-Hermitian Hamiltonians the eigenvectors of a given $H$ and ${H}^\dagger$ are, in fact, non-orthogonal but complete in their Hilbert space. What may or may not be true is that they are also bases for such a Hilbert space~\cite{bagbook}. Hence, it is surely worth to investigate this kind of properties for  $\F_\Phi$ and $\F_\Psi$. In the present case, we obtain the following interesting result:

\begin{prop}
If $V$ is such that $V^2$ is not a natural number, then $\F_\Phi$ and $\F_\Psi$ are complete in $\Hil_2$. If, on the other hand, $V^2$ is a positive integer number $m_0$, then $\F_\Phi$ and $\F_\Psi$ are not complete.
\end{prop}
\begin{proof}
Let $f=
                  \begin{pmatrix}
                    f_1 \\
                    f_2
                  \end{pmatrix}
                $ be a vector which is orthogonal to all the eigenvectors $\Phi_{n_1,n_2}^{(\pm)}$. We would like to show if and in which condition $f$ is zero.

First of all, since $\left\langle f,\Phi_{n_1,0}^{(+)}\right\rangle_2=0$ in particular, it follows that $\left\langle f_1,e_{n_1,0}\right\rangle=0$ for all $n_1\geq0$. Moreover, we also have, for $n_2\geq1$ and for all $n_1$,
\begin{align}
0=\left\langle f,\Phi_{n_1,n_2}^{(\pm)}\right\rangle_2=\left\langle f_1,e_{n_1,n_2}\right\rangle+\alpha_{n_1,n_2}^{(\pm)}\left\langle f_2,e_{n_1,n_2-1}\right\rangle.
\label{310}
\end{align}
Then, by subtraction, we have $\left(\alpha_{n_1,n_2}^{(+)}-\alpha_{n_1,n_2}^{(-)}\right)\left\langle f_2,e_{n_1,n_2-1}\right\rangle=0$ but, since
$\alpha_{n_1,n_2}^{(+)}-\alpha_{n_1,n_2}^{(-)}=-2i\sqrt{n_2-V^2}/\sqrt{n_2}$, it follows that, if $V^2$ is not equal to any natural numbers, then $\left\langle f_2,e_{n_1,n_2-1}\right\rangle=0$ for all $n_1$ and for all $n_2\geq1$. Then, because of the completeness of the set $\E$, we conclude that $f_2=0$. This result, together with (\ref{310}), now implies that $\left\langle f_1,e_{n_1,n_2}\right\rangle=0$ for all $n_1$ and for $n_2\geq1$. Since we also have that $\left\langle f_1,e_{n_1,0}\right\rangle=0$, however, it follows that $f_1=0$ after using again the completeness of $\E$. Hence $f=0$.

\vspace{2mm}

Let us now check what happens if $V^2=m_0$, for some particular natural number $m_0$. In this case we find that
$
\alpha_{n_1,m_0}^{(+)}=\alpha_{n_1,m_0}^{(-)}=-1
$
for all $n_1$, and therefore
\begin{align}
\Phi_{n_1,m_0}^{(+)}=\Phi_{n_1,m_0}^{(-)}=\frac{1}{\sqrt{2}}
                  \begin{pmatrix}
                    e_{n_1,m_0} \\
                    -e_{n_1,m_0-1}
                  \end{pmatrix}.
\end{align}
We see that we are {\em losing one vector}, so that it is not really surprising that the set $\F_\Phi$ ceases to be complete. In fact, a simple computation shows that, for instance, the non-zero vector $
                  \begin{pmatrix}
                    e_{0,m_0} \\
                    e_{0,m_0-1}
                  \end{pmatrix}
                $ is orthogonal to all the eigenvectors $\Phi_{n_1,n_2}^{(\pm)}$ as well as to $
                  \begin{pmatrix}
                    e_{n_1,m_0} \\
                    e_{n_1,m_0-1}
                  \end{pmatrix}
                $ for all fixed $n_1\geq0$.

\vspace{1mm}

A similar proof can be repeated for the set $\F_\Psi$.

\end{proof}

\vspace{2mm}

{\bf Remarks:} (i) The content of this Proposition can be understood in terms of exceptional points: we have an exceptional point when $V^2=m_0$, for a natural number $m_0$, while no exceptional point exists if $V^2$ is not natural. It is exactly the presence of an exceptional point which makes two eigenvectors collapse into a single one, and this prevent $\F_\Phi$ to be a basis.

(ii) The above result implies that in order for $\F_\Phi$ or $\F_\Psi$ to be bases for $\Hil_2$, $V$ must be such that its square is not a natural number, since any basis must be, first of all, complete and, in this situation, our sets are not. On the other hand, whenever $V^2$ is not an integer, $\F_\Phi$ or $\F_\Psi$ could be bases, but the question is, for the time being, still open.  { We believe that, even if this is often not so for non-Hermitian Hamiltonians~\cite{bagbook}, it is probably true in the present situation.}

\subsection{The $K'$ Dirac cone}
\label{sec3.2}

It is now interesting to observe that the results that we have deduced so far can be easily adapted to the other Dirac cone at $K'$. This is because the Hamiltonian $H_{K'}^{(+)}(V)$ in this case is simply the transpose of $H_K^{(+)}(V)$ in Eq.~(\ref{31}). Hence we have
\begin{align}
H_{K'}^{(+)}(V)={H_{K}^{(+)}}^T(V)=\frac{2iv_F}{\xi}
      \begin{pmatrix}
        V & -A_2 \\
        A_2^\dagger & -V \\
      \end{pmatrix}.
\end{align}
If we now compare the generic eigenvalue equations for $H_{K}^{(+)}(V)$ and $H_{K'}^{(+)}(V)$,
\begin{align}
H_{K}^{(+)}(V)
                  \begin{pmatrix}
                    \varphi_1 \\
                    \varphi_2
                  \end{pmatrix}
                =E
                  \begin{pmatrix}
                    \varphi_1 \\
                    \varphi_2
                  \end{pmatrix},
               \qquad \mbox{and }\qquad H_{K'}^{(+)}(V)
                  \begin{pmatrix}
                    \varphi_1' \\
                    \varphi_2'
                  \end{pmatrix}
                =E'
                  \begin{pmatrix}
                    \varphi_1' \\
                    \varphi_2'
                  \end{pmatrix},
\end{align}
it is easy to see that the second equation is mapped into the first one if we put $\varphi_1'=\varphi_2$, $\varphi_2'=-\varphi_1$ and $E'=-E$. Hence the conclusion is that the eigenvectors of $H_{K'}^{(+)}(V)$ are just those which we have deduced previously after this changes, and that the eigenvalues are just those of $H_{K}^{(+)}(V)$ but with signs exchanged. More in details we find that, for all $n_1\geq0$ and $n_2\geq1$,
\begin{align}
H_{K'}^{(+)}(V){\Phi'}_{n_1,n_2}^{(\pm)}={E'}_{n_1,n_2}^{(\pm)}{\Phi'}_{n_1,n_2}^{(\pm)},
\label{311}
\end{align}
where ${E'}_{n_1,n_2}^{(\pm)}=-{E}_{n_1,n_2}^{(\pm)}={E}_{n_1,n_2}^{(\mp)}$ and
\begin{align}
{\Phi'}_{n_1,n_2}^{(\pm)}=\frac{1}{\sqrt{1+|\alpha_{n_1,n_2}^{(\pm)}|^2}}
                  \begin{pmatrix}
                    \alpha_{n_1,n_2}^{(\pm)} e_{n_1,n_2-1} \\
                    - e_{n_1,n_2}
                  \end{pmatrix}.
\label{312}
\end{align}

When $n_2=0$ we have
\begin{align}
H_{K'}^{(+)}(V){\Phi'}_{n_1,0}^{(+)}={E'}_{n_1,0}^{(\pm)}{\Phi'}_{n_1,0}^{(+)},
\label{313}
\end{align}
where ${E'}_{n_1,0}^{(+)}={E}_{n_1,0}^{(-)}= - 2iv_FV/\xi$, and
\begin{align}
{\Phi'}_{n_1,+}^{(+)}=
                  \begin{pmatrix}
                    0 \\
                    - e_{n_1,0}
                  \end{pmatrix}.
\label{314}
\end{align}
Combining ${E'}_{n_1,0}^{(+)}= - 2iv_FV/\xi$ for $H_{K'}^{(+)}(V)$ with $E_{n_1,0}^{(+)}=2iv_FV/\xi$ for $H_{K}^{(+)}(V)$ (see below Eq.~\eqref{33}), we see that these two eigenvalues become complex as in Fig.~\ref{fig-exceptional point} but without the horizontal arrows.

Notice that, similarly to what happened for $H_{K}^{(+)}(V)$, the Hamiltonian $H_{K'}^{(+)}(V)$ has no (non-zero) eigenstate corresponding to ${E'}_{n_1,0}^{(-)}$.
Similar features as those considered for $H_{K}^{(+)}(V)$ arise also here, as for instance the completeness of the sets of eigenstates of $H_{K'}^{(+)}(V)$ and of its adjoint, and the conclusions do not differ from what we have found so far; we will not repeat similar considerations here.

\section{Perspectives and conclusion}
\label{sec5}

In this paper we have considered an extended non-Hermitian version of the graphene Hamiltonian close to the Dirac points  $K$ and $K'$. On a mathematical side we have shown that, depending on the value of the parameter $V$ measuring this non-Hermiticity, exceptional points may arise, which breaks down the existence of a basis for $\Hil_2$. In fact, the set of eigenstates of $H_{K}^{(+)}(V)$ is not even complete at the exceptional points. We have also deduced an interesting behavior concerning the zeroth eigenvalues and eigenvectors of the model: while $\Phi_{n_1,0}^{(+)}$ does exist, no $\Phi_{n_1,0}^{(-)}$ can be found in $\Hil_2$, at least if $V\neq0$. Similarly, $\Psi_{n_1,0}^{(-)}$ does exist, but $\Psi_{n_1,0}^{(+)}$ does not. Hence, introducing $V$ in the Hamiltonian creates a sort of asymmetry between the plus and the minus eigenstates, at least for the ground state. This asymmetry disappears as soon as $V$ is sent to zero.

If we compare the conclusion for the $\cPT$-symmetric graphene with the physical view of the simplest case that we described in Introduction, we may say the following.
The electrons doped on one sublattice may not be carried to the other sublattice through the central channels $n_2=0$ as soon as we introduce the $\cPT$-symmetric chemical potential.
The other channels remain open until $n_2=V^2$, when the corresponding $n_2$th channel is closed.
It may be an interesting future work to drive the system around an exceptional point to see the state swapping~\cite{Berry11,Gilary13,Milburn15}.

\section*{Acknowledgements}

This work was supported by National Group of Mathematical Physics (GNFM-INdAM). F.B. also acknowledges partial support by the University of Palermo.

\section*{Computational solution}

This paper does not contain any computational solution.

\section*{Ethics statement}

This work did not involve any active collection of human data.

\section*{Data accessibility statement}

This work does not have any experimental data.

\section*{Competing interests statement}

We have no competing interests.

\section*{Authors' contributions}

FB cured the mathematical part of the paper, with the help of NH. NH cured the physical interpretation of the results, with the help of FB. Both authors gave final approval for publication.

\section*{Funding}

This work was partly supported by GNFM-INdAM and by the University of Palermo.

\appendix

\renewcommand{\thesection}{Appendix \Alph{section}}
\renewcommand{\theequation}{\Alph{section}.\arabic{equation}}

\section{Some general facts for non-Hermitian Hamiltonians}
\label{appA}

We here briefly describe the general notion of the intertwining operator that we have introduced in Sec.~\ref{sec3.1}.
In order to avoid mathematical problems, we focus here on finite-dimensional Hilbert spaces. In this way our operators are finite matrices.

The main ingredient is an operator (i.e.\ a matrix) $H$, acting on the vector space ${\Bbb C}^{N+1}$, with $H\neq H^\dagger$ and with exactly $N+1$ distinct eigenvalues $E_n$, $n=0,1,2,\ldots,N$, where the Hermitian conjugate $H^\dagger$ of $H$ is the usual one, i.e.\ the complex conjugate of the transpose of the matrix $H$. Because of what follows, and in order to fix the ideas, it is useful to remind here that the Hermitian conjugate $X^\dagger$ of an operator $X$ is defined in terms of the {\em natural} scalar product $\left\langle .,.\right\rangle$ of the Hilbert space $\Hil=\left({\Bbb C}^{N+1},\left\langle .,.\right\rangle\right)$:
$\left\langle Xf,g\right\rangle=\left\langle f,X^\dagger g\right\rangle$, for all $f,g\in {\Bbb C}^{N+1}$, where $\left\langle f,g\right\rangle=\sum_{k=0}^N\overline{f_k}\,g_k$, with obvious notation.

In this Appendix we will restrict to the case in which all the eigenvalues $E_n$ are real, and with multiplicity one. Hence
\begin{align}
H\varphi_k=E_k\varphi_k.
\label{a1}
\end{align}
The set $\F_\varphi=\{\varphi_k,\,k=0,1,2,\ldots,N\}$ is a basis for ${\Bbb C}^{N+1}$, since the eigenvalues are all different. Then an unique biorthogonal basis of $\Hil$, $\F_\Psi=\{\Psi_k,\,k=0,1,2,\ldots,N\}$, surely exists~\cite{you,chri}: $\left\langle \varphi_k,\Psi_l\right\rangle=\delta_{k,l}$, for all $k, l$. It is easy to check that $\Psi_k$ is automatically an eigenstate of $H^\dagger$, with eigenvalue $E_k$:
\begin{align}
H^\dagger\Psi_k=E_k\Psi_k.
\label{a2}
\end{align}
Using the bra-ket notation we can write $\sum_{k=0}^N|\varphi_k\left\rangle \right\langle \Psi_k|=\sum_{k=0}^N|\Psi_k\left\rangle \right\langle \varphi_k|=\1$, where, for all $f,g,h\in\Hil$, we define $(|f\left\rangle \right\langle g|)h:=\left\langle g,h\right\rangle f$.

We now introduce the `intertwining' operators $S_\varphi=\sum_{k=0}^N|\varphi_k\left\rangle \right\langle \varphi_k|$ and $S_\Psi=\sum_{k=0}^N|\Psi_k\left\rangle \right\langle \Psi_k|$, following Ref.~\cite{bagbook}. These are bounded positive, Hermitian, invertible operators, one the inverse of the other: $S_\Psi=S_\varphi^{-1}$.

Moreover
 \begin{align}
S_\varphi\Psi_n=\varphi_n,\quad S_\Psi\varphi_n=\Psi_n,
\label{a3}
\end{align}
and we also get the following intertwining relations involving $H$, $H^\dagger$, $S_\varphi$ and $S_\Psi$:
\begin{align}
S_\Psi H=H^\dagger S_\Psi,\quad S_\varphi H^\dagger= HS_\varphi.
 \label{a4}
 \end{align}
Notice that the second equality follows from the first one, by left and right multiplying $S_\Psi H=H^\dagger S_\Psi$ with $S_\varphi$. To prove the first equality, we first observe that $(S_\Psi H-H^\dagger S_\Psi)\varphi_n = 0$ for all $n$. Hence our claim follows because of the basis nature of $\F_\varphi$.

\vspace{2mm}

{\bf Remark:} It might be interesting to recall that the intertwining operators, such as $S_\varphi$ and $S_\Psi$, are quite useful in quantum mechanics, $\cPT$-symmetric or not \cite{intop1,intop2,intop3,intop4,intop5}, in order to deduce eigenvectors of certain Hamiltonians connected by intertwining relations. For instance, let us assume that $\varphi_n$ is an eigenstate of a certain operator $H_1$ with eigenvalue $E_n$: $H_1\varphi_n=E_n\varphi_n$, and let us also assume that two other operators $H_2$ and $X$ exist such that $\varphi_n\notin\ker(X)$ and that the intertwining relation $XH_1=H_2X$ is satisfied. This is exactly what happens in (\ref{a4}), identifying $X$ with $S_\Psi$, $H_1$ with $H$ and $H_2$ with $H^\dagger$.

Then, it is a trivial exercise to check that the non-zero vector $\Psi_n=X\varphi_n$ is an eigenstate of $H_2$, with eigenvalue $E_n$. Indeed we have
$$
H_2\Psi_n=H_2\left(X\varphi_n\right)=XH_1\varphi_n=X\left(E_n\varphi_n\right)=E_n X\varphi_n=E_n \Psi_n.
$$
Note that the fact that $H_1$ and $H_2$ are Hermitian or not and the fact that $E_n$ is real or not play no role. Note also that the fact that $\Psi_n$ can be deduced out of $\varphi_n$ simply by applying $X$, is exactly what happens in our situation; see Eq.~(\ref{a3}). This explains why the intertwining operators are so important in concrete applications; they can be used, for instance, to find eigenstates of new operators starting from eigenstates of old ones.

\vspace{2mm}

\vspace{2mm}

{\bf Remark:} It is probably worth mentioning that not all we have discussed here can be easily extended if $\dim(\Hil)=\infty$. For instance, considering the intertwining relations in (\ref{a4}), if, for instance, $H$ and $S_\varphi$ are unbounded, taken $f\in D(S_\varphi)$, the domain of $S_\varphi$, there is no reason \textit{a priori} for $S_\varphi f$ to belong to $D(H)$, so that $H S_\varphi f$ needs not to be defined.

\section{$\cT$ and $\cP$ symmetries of the model with the $\cPT$-symmetric potential}
\label{appB}

In this Appendix we will briefly discuss the role of the $\cT$ and $\cP$ symmetries in our model. The $\cT$ operator works as follows:
\begin{align}
&\cT x\cT=x,\quad
\cT y\cT=y,\quad
\cT p_x\cT=-p_x,\quad
\cT p_y\cT=-p_y,
\\
&\cT i \cT=-i,\quad
\cT B\cT=-B;
\end{align}
note that, as expected for physical reasons, the time-reversal operator flips the magnetic field too.
We therefore have
\begin{align}
&\cT a_X \cT = a_X,\quad
\cT a_Y \cT = a_Y,\quad
\\
&\cT A_1 \cT =A_2,\quad
\cT A_2 \cT=A_1.
\end{align}
When we apply $\cT$ to
\begin{align}
H_K^{(+)}(V)=\frac{2iv_F}{\xi}
\begin{pmatrix}
V & A_2^\dag \\
-A_2 & -V
\end{pmatrix},
\end{align}
we have
\begin{align}\label{eq190}
\cT H_K^{(+)} (V)\cT=
-\frac{2iv_F}{\xi}
\begin{pmatrix}
V & A_1^\dag \\
-A_1 & -V
\end{pmatrix}
=-H_{K'}^{(-)}(V).
\end{align}

We thus realize that the time reversal of the Dirac cone at $K$ is the negative of the Dirac cone at $K'$.
For $V=0$, we can flip the sign by the diagonal unitary transformation
\begin{align}
\mathcal{U}=\begin{pmatrix}
1 & 0 \\
0 & -1
\end{pmatrix}
\end{align}
as in $(\cT \mathcal{U})H_K^{(+)}(0)(\cT \mathcal{U})=H_{K'}^{(-)}(0)$.
Similarly we have $(\cT \mathcal{U})H_{K'}^{(+)}(0)(\cT \mathcal{U})=H_{K}^{(-)}(0)$.
Note  that $H_K^{(+)}$ and $H_K^{(-)}$ are different expressions of the same Hamiltonian~\eqref{20}, expressions which depend on the direction of the magnetic field along $z$.
The model for $V=0$ is time-reversal symmetric in this sense.
Under $\cT$, the Dirac cone at $K$ is transformed to the one at $K'$, which in turn is transformed to the one at $K$.
Therefore, the set of the two Dirac cones for $V=0$ has the time-reversal symmetry.

The time-reversal symmetry is broken when $V\neq 0$ because $(\cT \mathcal{U})H_K^{(+)}(V)(\cT \mathcal{U})=H_{K'}^{(-)}(-V)\neq H_{K'}^{(-)}(V)$.
This is also true for the parity operation
\begin{align}
\cP=\begin{pmatrix}
0 & 1 \\
1 & 0
\end{pmatrix},
\end{align}
for which we have
\begin{align}
\cP H_K^{(+)}(V) \cP=\frac{2iv_F}{\xi}
\begin{pmatrix}
-V & -A_2 \\
A_2^\dag & V
\end{pmatrix}.
\end{align}
For $V=0$, this is isomorphic to $H_{K'}^{(-)}$ but for $V\neq 0$, $\cP H_K^{(+)}(V) \cP$ is isomorphic to $H_{K'}^{(-)}(-V)$, which is not equal to $H_{K'}^{(-)}(V)$.

For $V\neq 0$, however, the $\cPT$ symmetry is satisfied:
\begin{align}
(\cT \cP \mathcal{U})H_{K}^{(+)}(V)(\cT\cP\mathcal{U})
&=\cP H_{K'}^{(-)}(-V) \cP
\\
&=\frac{2iv_F}{\xi}
\begin{pmatrix}
V & -A_1 \\
A_1^\dag & -V
\end{pmatrix}
\\
&=H_{K'}^{(-)}(V).
\end{align}
Therefore, for $V\neq0$, the $\cT$ and $\cP$ symmetries are broken but $\cPT$ symmetry is not.

\end{document}